\documentclass[a4paper,12pt]{article}
\usepackage{amssymb,amsfonts,graphicx,amsmath,amsthm,epsfig,tabularx}
\usepackage{float}
\restylefloat{table}
\usepackage{color}
\newtheorem{theorem}{Theorem}[section]
\newtheorem{lemma}[theorem]{Lemma}
\newtheorem{proposition}[theorem]{Proposition}

\title{The structure of graphs with given number of blocks and the maximum Wiener index}

\author{St\'ephane Bessy\footnote {Laboratoire d'Informatique, de Robotique et de
  Micro\'{e}lectronique de Montpellier (LIRMM), Universit\'e de
  Montpellier, France, {\tt stephane.bessy@lirmm.fr}.},
Fran\c cois Dross\footnote {Laboratoire d'Informatique, de Robotique et de
  Micro\'{e}lectronique de Montpellier (LIRMM), Universit\'e de
  Montpellier, France, {\tt francois.dross@lirmm.fr}.},
Katar{\' i}na~Hri{\v n}{\'a}kov{\'a}\footnote{Department of Mathematics, 
  Faculty of Civil Engineering, Slovak University of Technology in Bratislava, 
  Radlinsk{\'e}ho 11, 810 05, Bratislava, Slovakia, {\tt hrinakova@math.sk}.},\\
Martin Knor\footnote{Department of Mathematics, 
  Faculty of Civil Engineering, Slovak University of Technology in Bratislava, 
  Radlinsk{\'e}ho 11, 810 05, Bratislava, Slovakia, {\tt knor@math.sk}.},
Riste~{\v S}krekovski\footnote{Faculty of Information
  Studies, 8000 Novo Mesto 
  \& Faculty of Mathematics and Physics, University of Ljubljana, 1000 Ljubljana 
  \& FAMNIT, University of Primorska, 6000 Koper, Slovenia,
  {\tt skrekovski@gmail.com}.}}

\begin{document}

\maketitle

\begin{abstract}
The Wiener index (the distance) of a connected graph is the sum of
distances between all pairs of vertices. In this paper, we study the
maximum possible value of this invariant among graphs on $n$ vertices
with fixed number of blocks $p$. It is known that among graphs on $n$
vertices that have just one block, the $n$-cycle has the largest
Wiener index.  And the $n$-path, which has $n-1$ blocks, has the
maximum Wiener index in the class of graphs on $n$ vertices.  We show
that among all graphs on $n$ vertices which have $p\ge 2$ blocks, the
maximum Wiener index is attained by a graph composed of two cycles
joined by a path (here we admit that one or both cycles can be
replaced by a single edge, as in the case $p=n-1$ for example).

\end{abstract}





\section{Introduction}

Let $G$ be a simple graph.
By $V(G)$ and $E(G)$ we denote the vertex set and the edge set of $G$,
respectively.
Let $u$ and $v$ be two vertices of $G$.
The length of a shortest $u{-}v$ path is denoted by $d_G(u,v)$,
or simply by $d(u,v)$ if no confusion is likely.
The Wiener index is defined as the sum of the distances between
all (unordered) pairs of vertices of $G$,
$$
W(G)=\sum_{\{u,v\}\subseteq V(G)} d(u,v).
$$
The {\em transmission} of a vertex $v$ is the sum of the distances from $v$
to other vertices of $G$, i.e.,
$w_G(v)=\sum_{u\in V(G)} d_G(u,v)$.
Then the Wiener index of $G$ equals
$\frac{1}{2}\sum_{u\in G} w_G (u)$.

The Wiener index was introduced by Wiener \cite{Wiener} in 1947,
thus it is one of the oldest topological descriptors.
At first it was used for predicting the boiling points of paraffins,
later some other applications of the Wiener index were revealed.
Many years later it was studied also from a purely graph-theoretical
point of view.
But mathematicians studied the Wiener index under different
names, such as the gross status \cite{Harary}, the distance of a graph
\cite{EJS} and the transmission \cite{Soltes}.
More details can be found in some of the many surveys, see e.g.
\cite{surv1,ma4,ma5,Gut-sur}.

If $G$ is a connected graph and $v$ is a cut-vertex that partitions
$G$ into subgraphs $G_1$ and $G_2$, i.e., $G=G_1\cup G_2$ and
$G_1\cap G_2 = \{v\}$, then we write $G=G_1\circ_{v}G_2$,
or simply $G=G_1\circ G_2$.
By $C_n$, $P_n$ and $K_n$ we denote a cycle, path and a complete graph,
respectively, on $n$ vertices.
We will abuse this notation by writing $C_2=K_2$.
Then our main result is the following statement.

\begin{theorem}
\label{thm:main}
Let $n$ and $p$ be numbers such that $n>p>1$.
Among all graphs on $n$ vertices with $p$ blocks, the maximum
Wiener index is attained by the graph
$C_a\circ_u P_{p-1}\circ_v C_b$
for some integers $a \ge 2$ and $b \ge 2$, where $a+b=n-p+3$,
and $u$ and $v$ are distinct endvertices of $P_{p-1}$.
\end{theorem}

Note that $C_a$ or $C_b$ can also be $C_2$, i.e.~an edge, and then we
obtain $C_{n-p+1}\circ_u P_{p}$, which is a graph composed of one cycle
with an attached path. In the case when $p=n-1$, both $C_a$ and $C_b$ are
edges, i.e.~$a=b=2$.

The proof of Theorem~\ref{thm:main} is rather technical. Therefore, 
the exact values of $a$ and $b$ will be determined in a forthcoming paper~\cite{BDHKS}.
Let $W_n(p)$ be the maximum Wiener index of a graph which has $n$ vertices and $p$ blocks.
In~\cite{BDHKS} we study $W_n(p)$ and we determine its minimum values.

\medskip

Now, we introduce notations and definitions which we use throughout
the paper.  If $G,G_0,G_1,\dots$ are graphs, we denote by
$n,n_0,n_1,\dots$, respectively, their numbers of vertices.  For $v\in
V(G)$, by $e_G(v)$ we denote the {\em eccentricity} of $v$ in $G$,
i.e., the maximal distance from $v$ in $G$.

A graph is {\em nonseparable} if it is connected and has no
cut-vertices (i.e.~either it is $2$-connected or it is $K_2$).
A {\it block} of $G$ is a maximal non-separable subgraph of $G$.
Two blocks sharing a common vertex are said to be {\em adjacent}.
We refer to~\cite{BM08} concerning the structure of blocks in a 
connected graph.
In particular, it is known that the bipartite graph built on the set
of blocks of $G$ and the set of cut-vertices of $G$ by linking
a block to the cut-vertices it contains, is a tree.
This tree is called the {\em blocks-tree} of $G$.

Let $H$ be a subgraph of $G$, such that $H$ is a connected union of several
(at least one) blocks of $G$.
An {\em attachment vertex} of $H$ is a vertex of $H$ which has a neighbour
in $G\setminus H$.
The subgraph $H$ is {\em terminal} if $H$ contains exactly one attachment
vertex.
It is {\em traversal} if it contains exactly two attachment vertices.

Let $G$ be a connected graph on $n$ vertices.
The {\em distance vector} of a vertex $v$ is the $e_G(v)$-dimensional
vector $d_G(v)$ given by $d_G(v)_i=|\{x\in G\,:\ d_G(v,x)=i\}|$.
If $\omega$ is a vector, then $\langle\omega\rangle$ is the value
$\sum_{i} i \omega_i$.
Observe that $\langle d_G(v) \rangle=w_G(v)$.

Now we define ${\bf 2}_n$.
If $n$ is even, the vector ${\bf 2}_n$ has dimension $n/2$ and
contains the value $2$ in each coordinate except for the last one
which is 1.
If $n$ is odd, ${\bf 2}_n$ has dimension $(n-1)/2$ and each of
its coordinates has value 2.
For example ${\bf 2}_7=(2,2,2)$ and ${\bf 2}_6=(2,2,1)$.
Observe that the vectors $d_{C_n}(x)$ and ${\bf 2}_n$ are the same for every
vertex $x$ of the cycle $C_n$.
Hence we obtain $w_{C_n}(x)=\frac{n^2}4$ if $n$ is even and
$w_{C_n}(x)=\frac{n^2-1}4$ if $n$ is odd.
Also observe that if $G$ is a $2$-connected graph, then the distance
vector of every vertex $v$ of $G$ satisfies $d_G(v)_i\ge 2$ for every
$i<e_G(v)$, and so $w_G(v)\le \langle {\bf 2}_n\rangle$.
Moreover, if $G$ is different from a cycle then it has a vertex $u$, such
that $d_G(u)_1\ge 3$, which means that
$w_G(u)<\langle{\bf 2}_n\rangle$.
So we get the following classical result.

\begin{proposition}
\label{prop:C_n}
Let $G$ be a $2$-connected graph on $n$ vertices and let $v\in V(G)$.
Then
$$
w_G(v)\le
\begin{cases}
\frac{n^2}4 \qquad\quad\mbox{if $n$ is even;}\\
\frac{n^2-1}4 \qquad\mbox{if $n$ is odd.}
\end{cases}
$$
Moreover, if $G$ is $C_n$ then equality holds for every vertex $v\in V(C_n)$.
Further, the cycle $C_n$ is the unique graph which has the maximal
Wiener index over the class of 2-connected graphs on $n$ vertices, and
$$
W(C_n)=
\begin{cases}
\frac{n^3}8 \qquad\quad\mbox{if $n$ is even;}\\
\frac{n^3-n}8 \qquad\mbox{if $n$ is odd.}
\end{cases}
$$
\end{proposition}

We use also the following obvious statement.

\begin{proposition}
\label{prop:P_n}
Let $v$ be an endvertex of $P_n$.
Then
$$
w_{P_n}(v)=\binom n2
\qquad\mbox{and}\qquad
W(P_n)=\binom{n+1}3.
$$
\end{proposition}


\section{Proof of Theorem~{\ref{thm:main}}}

In this section we prove Theorem~{\ref{thm:main}} using a couple
of auxiliary results.
The first two propositions will be useful to calculate Wiener
index of a~graph composed of two or more subgraphs joined by cut-vertices.
The proofs are straightforward, so we omit them.
Recall that the number of vertices of $G_i$ is denoted by $n_i$.

\begin{proposition}
\label{prop:W1}
Let $G=G_1\circ_vG_2$.
We have
$$
W(G) = W(G_1) + W(G_2) + w_{G_1}(v)\cdot(n_2-1) + w_{G_2}(v)\cdot(n_1-1).
$$
\end{proposition}

Observe that the subgraphs $G_1$ and $G_2$ in the previous proposition do
not need to be blocks.
In fact, each of these graphs is either a block or a connected union of
blocks of $G$.
Using an inductive argument we can get the following generalization of
Proposition~{\ref{prop:W1}}

\begin{proposition}
\label{prop:W2}
Let $G_1,G_2,\dots,G_{\ell}$ be blocks or connected unions of blocks of $G$,
such that $E(G_1),E(G_2),\dots,E(G_{\ell})$ is an edge decomposition of
$E(G)$.
Denote by $v_{i,j}$ the attachment vertex of $G_i$ which separates
$G_i\setminus\{v_{i,j}\}$ from $G_j\setminus\{v_{i,j}\}$.
Then
\begin{align}
&W(G)=\sum_{i=1}^{\ell} W(G_i)\,+\nonumber\\
\label{sum:W2}
&\sum_{1\le i<j\le\ell}
\left(w_{G_i}(v_{i,j})\cdot(n_j{-}1) + w_{G_j}(v_{j,i})\cdot(n_i{-}1)
+d_G(v_{i,j},v_{j,i})\cdot(n_i{-}1)\cdot(n_j{-}1)\right).
\end{align}
\end{proposition}

Observe that the last term in the second sum of Proposition~{\ref{prop:W2}}
is $0$ if $G_i$ and $G_j$ are adjacent blocks.
We remark that Proposition~{\ref{prop:W2}} holds even in the case when some
of the $G_i$ are ``trivial", i.e., if they consist of a single vertex, since
then all the terms containing $W(G_i)$, $(n_i-1)$, or $w_{G_i}(v_{i,j})$ are
zeros.

Now we show that terminal blocks are cycles or edges in extremal graphs.

\begin{lemma}
\label{lem:terminal-blocks}
Let $B$ be a terminal block of $G$ such that $B$ is not a cycle and
$|V(B)|\ge 3$.
Let $G'$ be the graph obtained from $G$ by replacing $B$
by a cycle on $|V(B)|$ vertices.
Then $W(G')>W(G)$.
\end{lemma}

\begin{proof}
Denote by $v$ the attachment vertex of $B$ in $G$.
Further, denote by $G_1$ the block $B$ and denote by $G_2$ the subgraph of
$G$ such that $G_1\circ_v G_2=G$.
By Proposition~{\ref{prop:W1}} we have (recall that $n_i=|V(G_i)|$)
\begin{align*}
W(G)&=W(B)+W(G_2)+w_B(v)\cdot(n_2-1)+w_{G_2}(v)\cdot(n_1-1)\\
W(G')&=W(C_{n_1})+W(G_2)+w_{C_{n_1}}(v)\cdot(n_2-1)+w_{G_2}(v)\cdot(n_1-1)
\end{align*}
and so
$$
W(G')-W(G)=W(C_{n_1})-W(B)+\big(w_{C_{n_1}}(v)-w_B(v)\big)\cdot(n_2-1).
$$
Since $B$ is not a cycle, we have $W(C_{n_1})-W(B)>0$ by
Proposition~{\ref{prop:C_n}}.
Moreover, by Proposition~{\ref{prop:C_n}} we have also
$w_{C_{n_1}}(v)=\langle{\bf 2}_{n_1}\rangle\ge\langle d_B(v)\rangle=w_B(v)$.
Hence we obtain $W(G')>W(G)$.
\end{proof}

In a cycle $C_n$, two vertices $u$ and $v$ are {\em opposite}
(or {\em antipodal}) if they satisfy
$d_{C_n}(u,v)=\max \{d_{C_n}(x,y):x,y\in C\}=\lfloor\frac n2\rfloor$.

\begin{lemma}
\label{lem:traversal-blocks}
Let $B$ be a traversal block of $G$ with $|V(B)|=n_0\ge 3$, and let $v_1$
and $v_2$ be the two attachment vertices of $B$.
Let $C_{n_0}$ be a cycle in which $v_1$ and $v_2$ are opposite and let
$G'$ be obtained from $G$ by replacing $B$ by $C_{n_0}$.
If $B$ is not a cycle or if $B$ is a cycle and $v_1$ and $v_2$ are not
opposite in $B$, then $W(G')>W(G)$.
\end{lemma}

\begin{proof}
Denote by $G_1$ and $G_2$ the subgraphs of $G$ attached to $B$ at $v_1$ and
$v_2$, respectively, such that $E(G)=E(G_1)\cup E(B)\cup E(G_2)$.
By Proposition~{\ref{prop:W2}} we have
\begin{align*}
&W(G')-W(G)=\big(W(C_{n_0})-W(B)\big)
+\big(w_{C_{n_0}}(v_1)-w_B(v_1)\big)\cdot(n_1-1)+\\
&\big(w_{C_{n_0}}(v_2)-w_B(v_2)\big)\cdot(n_2-1)
+\big(d_{C_{n_0}}(v_1,v_2)-d_B(v_1,v_2)\big)\cdot(n_1-1)\cdot(n_2-1).
\end{align*}
By Proposition~{\ref{prop:C_n}} we have $W(C_{n_0})-W(B)\ge 0$ and
equality holds if and only if $B$ is a cycle.
By Proposition~{\ref{prop:C_n}} we have also
$w_{C_{n_0}}(v_1)-w_B(v_1)\ge 0$ and $w_{C_{n_0}}(v_2)-w_B(v_2)\ge 0$.
Finally, since every vertex $v$ in a $2$-connected graph $H$ satisfies
$e_H(v)\le\big\lfloor\frac{|V(H)|}2\big\rfloor$ (recall that for every
$i<e_H(v)$ we have $d_H(v)_i\ge 2$), we have
$d_B(v_1,v_2)\le\lfloor\frac{n_0}2\rfloor=d_{C_{n_0}}(v_1,v_2)$.
Hence, all the terms on the right hand side of the equality for $W(G')-W(G)$
are nonnegative and they are all zeros if and only if $B=C_{n_0}$ and
$d_B(v_1,v_2)=\lfloor\frac{n_0}2\rfloor$.
\end{proof}

Next lemma gives a condition for extremal graphs.

\begin{lemma}
\label{lem:distance-2-vertices}
Let $G$ be a graph with at least $3$ blocks and let $G_1$ and $G_2$ be
two terminal cycles of $G$ with attachment vertices $v_1$ and $v_2$,
respectively.
Let $u_i$ be a vertex opposite to $v_i$ in $G_i$ for $i\in\{1,2\}$.
Denote by $G'$ (resp.~$G''$) a graph obtained from $G$ by removing the
block $G_1$ (resp.~$G_2$) and attaching it to $u_2$ (resp.~$u_1$).
Suppose that $W(G)\ge W(G')$ and $W(G)\ge W(G'')$.
Then $d_{G}(u_1,u_2)\ge\frac{n-1}2$.
\end{lemma}

\begin{proof}
Let $G_0$ be the graph obtained from $G$ by removing the cycles $G_1$
and $G_2$, such that $E(G)=E(G_1)\cup E(G_0)\cup E(G_2)$.
Observe that $G_0$ does not need to be a single block, but it is
a connected union of blocks.
Anyway, $G=G_1\circ_{v_1}G_0\circ_{v_2}G_2$,
$G'=G_0\circ_{v_2}G_2\circ_{u_2}G_1$ and
$G''=G_2\circ_{u_1}G_1\circ_{v_1}G_0$.
By Proposition~{\ref{prop:W2}} we have
\begin{align*}
W(G)-W(G')&=\big(w_{G_2}(v_2)-w_{G_2}(u_2)\big)\cdot(n_1-1)
+\big(w_{G_0}(v_1)-w_{G_0}(v_2)\big)\cdot(n_1-1)\\
&+d_{G_0}(v_1,v_2)\cdot(n_1-1)\cdot(n_2-1)
-d_{G_2}(v_2,u_2)\cdot(n_1-1)\cdot(n_0-1)
\end{align*}
where $w_{G_2}(v_2)=w_{G_2}(u_2)$.
Since $W(G)\ge W(G')$ and $n_1\ge 2$, we get
$$
w_{G_0}(v_1)-w_{G_0}(v_2)
+d_{G_0}(v_1,v_2)\cdot(n_2-1)
-d_{G_2}(v_2,u_2)\cdot(n_0-1)\ge 0.
$$
Analogously, from $W(G)-W(G'')\ge 0$ we get
$$
w_{G_0}(v_2)-w_{G_0}(v_1)
+d_{G_0}(v_1,v_2)\cdot(n_1-1)
-d_{G_1}(v_1,u_1)\cdot(n_0-1)\ge 0
$$
and summing the last two inequalities we obtain
$$
d_G(v_1,v_2)\cdot(n_1+n_2-2)
-\big(d_{G_2}(v_2,u_2)+d_{G_1}(v_1,u_1)\big)\cdot(n_0-1)\ge 0.
$$
Now since $v_i$ and $u_i$ are opposite in $G_i$ for $i\in\{1, 2\}$, we have
$$
d_G(v_1,u_1)+d_G(v_2,u_2)
=\Big\lfloor\frac{n_1}2\Big\rfloor+\Big\lfloor\frac{n_2}2\Big\rfloor
\ge \frac{n_1-1}2+\frac{n_2-1}2=\frac{n_1+n_2-2}2.
$$
Thus we obtain $d_G(v_1,v_2)\ge (n_0-1)/2$ and consequently
\begin{align*}
d_G(u_1,u_2)&=d_G(u_1,v_1)+d_G(v_1,v_2)+d_G(v_2,u_2)\\
&\ge\frac{n_1-1}2+\frac{n_0-1}2+\frac{n_2-1}2
=\frac{n-1}2
\end{align*}
since $n=n_1+n_0+n_2-2$.
\end{proof}

Let $n=tk+1$.
Take $k$ paths of length $t$ (i.e.~on $t+1$ vertices), on each path choose
one endvertex, and identify these endvertices.
We denote by $R^k_n$ the resulting graph.
Observe that $R^k_n$ has $n$ vertices and is homeomorphic to the star
$K_{1,k}$.
In \cite[Theorem~3]{KS} we have the following statement.

\begin{theorem}
\label{thm:k-dist}
Let $G$ be a connected graph on $n$ vertices.
Then for every $k$-tuple $u_1,u_2,\dots,u_k$ of its vertices, $3\le k<n$,
there are two, say $u_i$ and $u_j$ where $1\le i<j\le k$, such that
$$
d_G(u_i,u_j)\le\frac{2n-2}k.
$$
Moreover, if
$$
\min_{1\le i<j\le k} d_G(u_i,u_j)=\frac{2n-2}k
$$
then $n\equiv 1\pmod k$, the graph is $R^k_n$ and $u_1,u_2,\dots,u_k$ are the
endvertices of $R^k_n$.
\end{theorem}

Using Theorem~{\ref{thm:k-dist}} and Lemma~{\ref{lem:distance-2-vertices}}
we prove the following statement.

\begin{lemma}
\label{lema:<4}
Let $n>p$.
Let $G$ be a graph on $n$ vertices with $p$ blocks which has the maximum
Wiener index.
Then $G$ has at most three terminal blocks.
\end{lemma}

\begin{proof}
By way of contradiction, suppose that $G$ has at least four terminal blocks,
say $B_1$, $B_2$, $B_3$ and $B_4$.
By Lemma~{\ref{lem:terminal-blocks}} we know that each of these blocks is
either a cycle or $K_2$.
Let $v_i$ be the unique attachment vertex of $B_i$ and let $u_i$ be a
vertex opposite to $v_i$ in $B_i$, $1\le i\le 4$.
Denote
$$
d=\min_{1\le i<j\le 4} d_G(u_i,u_j)
$$
and assume that this minimum is attained by the pair $u_1,u_2$.
By Theorem~{\ref{thm:k-dist}} we know that $d\le\frac{n-1}2$.
We distinguish two cases.

{\bf Case~1:}
$d<\frac{n-1}2$.
Denote $G_1=B_1$ and $G_2=B_2$.
Now construct $G'$ and $G''$ by reattaching of $G_1$ and $G_2$ as in
Lemma~{\ref{lem:distance-2-vertices}}.
Since $d_G(u_1,u_2)=d<\frac{n-1}2$, either $W(G)<W(G')$ or $W(G)<W(G'')$.
Since all $G$, $G'$ and $G''$ have $n$ vertices and $p$ blocks, we get a
contradiction.

{\bf Case~2:}
$d=\frac{n-1}2$.
By Theorem~{\ref{thm:k-dist}}, in this case $G$ is $R^4_n$, and so $p=n-1$.
It is well-known that among trees on $n$ vertices, $P_n$ is the unique graph
with the maximum Wiener index.
So $W(P_n)>W(R^4_n)$, a contradiction.
\end{proof}

Now we prove some results useful for sequences of traversal blocks.
The following theorem was proved in~\cite{GCR14}.

\begin{theorem}
\label{thm:C.C}
For every $n\notin \{7,9\}$, the graph $C_{n-2}\circ C_3$ has the maximal
Wiener index among the graphs from the family $\{C_{n-r+1}\circ C_r : r\ge 3,\ n-r\ge 2\}$.
Moreover for $n=7$ and $n=9$, it holds $W(C_4\circ C_4)> W(C_5\circ C_3)$ and
$W(C_6\circ C_4) > W(C_7\circ C_3)> W(C_5\circ C_5)$.
\end{theorem}

We extend Lemma~{\ref{thm:C.C}} to blocks of size $2$.
(Recall that we denote the complete graph on $2$ vertices by $C_2$.)

\begin{lemma}
\label{lem:KC}
For every $n\geq 4$, among the graphs on $n$ vertices with exactly two
blocks, the maximal Wiener index is attained by $C_{n-1}\circ C_2$.
\end{lemma}

\begin{proof}
For $n=4$ the graph $C_3\circ C_2$ is the unique graph with two blocks,
thus it has the largest Wiener index.
For $n\ge 5$, $n\notin\{7,9\}$, it is enough to show that
$W(C_{n-1}\circ C_2)>W(C_{n-2}\circ C_3)$, by
Theorem~{\ref{thm:C.C}}.

Using Proposition~{\ref{prop:W1}} we get the Wiener index of
$G=C_{n-1}\circ_v K_2$.
\begin{align*}
W(G)&=W(C_{n-1})+W(K_2)+w_{C_{n-1}}(v)\cdot 1+w_{K_2}(v)\cdot(n-2)\\
&=\tfrac{n-1}2 \langle{\bf 2}_{n-1}\rangle+1
+\langle{\bf 2}_{n-1}\rangle\cdot 1+1\cdot(n-2).
\end{align*}
In $G'=C_{n-2}\circ_u C_{3}$ we can also use Proposition~{\ref{prop:W1}} to
evaluate the Wiener index.
\begin{align*}
W(G')&=W(C_{n-2})+W(C_3)+w_{C_{n-2}}(u)\cdot 2+w_{C_3}(u)\cdot(n-3)\\
&=\tfrac{n-2}2 \langle{\bf 2}_{n-2}\rangle+3
+\langle{\bf 2}_{n-2}\rangle\cdot 2+2\cdot (n-3).
\end{align*}
Hence, using Proposition~{\ref{prop:C_n}} we get
\begin{align*}
W(G)-W(G')&=\tfrac{n+1}{2}\langle{\bf 2}_{n-1}\rangle
-\tfrac{n+2}{2}\langle{\bf 2}_{n-2}\rangle-n+2\\
&=
\begin{cases}
\frac{(n-2)(n-4)}8 \qquad\ \ \quad\mbox{if $n$ is even;}\\
\frac{(n-1)(n-3)}8+1 \qquad\mbox{if $n$ is odd.}
\end{cases}
\end{align*}
Since $n\ge 5$, in both cases we get $W(G)>W(G')$.

By Theorem~{\ref{thm:C.C}}, for $n=7$ and $n=9$ it suffices to show
that $W(C_{n-1}\circ C_2)>W(C_{n-3}\circ C_4)$.
Direct computation gives $W(C_4\circ C_4)=40$, $W(C_6\circ C_2)=42$,
$W(C_6\circ C_4)=82$ and $W(C_8\circ C_2)=88$, which completes the proof.
\end{proof}

Using Lemma~{\ref{lem:KC}} we prove the following statement.
Here we allow the smaller end-block to be just a single vertex,
i.e.~$|V(G_0)|=1$, see below.

\begin{lemma}
\label{lem:replacing-two-cycles}
Let $G=G_0\circ_{v_1} G_1\circ_v G_2\circ_{v_2} G_3$, where $G_1$ and $G_2$
are cycles, $v_1$ and $v$ are antipodal in $G_1$, and $v$ and $v_2$ are
antipodal in $G_2$.
Let $k=n_1+n_2-1$, $n_0\le n_3$ and $n_3\geq 2$.
Then $G$ has maximal Wiener index if and only if
\begin{enumerate}
\item
$n_1=k-1$ and $n_2=2$, or 
\item
$n_1=2$, $n_2=k-1$ and $n_0=n_3$.
\end{enumerate}
\end{lemma}

\begin{proof}
Let $G'=G_0\circ_{v_1} C_{n_1+n_2-2}\circ_u C_2\circ_{v_2} G_3$,
where $u$ is antipodal to $v_1$ in $C_{n_1+n_2-2}$ and $u$ is antipodal to
(i.e., different from) $v_2$ in $C_2$ ($=K_2$).
Denote $H=C_{n_1}\circ C_{n_2}$ and $H'=C_{n_1+n_2-2}\circ C_2$.
By Proposition~{\ref{prop:W2}} we have
\begin{align*}
&W(G')-W(G)=\big(W(H')-W(H)\big)+
\big(w_{H'}(v_1)-w_H(v_1)\big)\cdot(n_0-1)\\
&+\big(w_{H'}(v_2)-w_H(v_2)\big)\cdot(n_3-1)
+\big(d_{H'}(v_1,v_2)-d_H(v_1,v_2)\big)\cdot(n_0-1)\cdot(n_3-1).
\end{align*}

By Lemma~{\ref{lem:KC}} we have $W(H')-W(H)\ge 0$ and equality holds if and
only if $H=H'$ (i.e., if $n_1=2$ or if $n_2=2$).
Further,
$d_{H'}(v_1,v_2)=\lfloor\frac{n_1+n_2-2}2\rfloor+\lfloor\frac 22\rfloor\ge
\lfloor\frac{n_1}2\rfloor+\lfloor\frac{n_2}2\rfloor=d_H(v_1,v_2)$, and so
the last term is nonnegative as well.
Let
$$
\Delta=\big(w_{H'}(v_1)-w_H(v_1)\big)\cdot(n_0-1)
+\big(w_{H'}(v_2)-w_H(v_2)\big)\cdot(n_3-1).
$$
We show that $\Delta\ge 0$.

If $k$ is even, we have
$w_{H'}(v_1)=\langle d_{H'}(v_1)\rangle=\langle(2,2,\ldots,2,1)\rangle$,
and $w_{H'}(v_2)=\langle d_{H'}(v_2)\rangle=\langle(1,2,2,\ldots,2)\rangle$,
where these vectors both have dimension $k/2$.
If $k$ is odd, we get
$w_{H'}(v_1)=\langle d_{H'}(v_1)\rangle=\langle(2,2,\ldots,2,1,1)\rangle$,
and $w_{H'}(v_2)=\langle d_{H'}(v_2)\rangle=\langle(1,2,2,\ldots,2,1)\rangle$,
where both these vectors have dimension $(k+1)/2$. 
To compute $w_{H}(v_1)$, $w_{H}(v_2)$ and $\Delta$ we distinguish four
cases according to the parity of $k$ and $n_1$.

{\bf Case~1:}
{\it Both $k$ and $n_1$ are even.}
Then $n_2$ is odd and
$w_H(v_1)=\langle d_H(v_1)\rangle=\langle(2,\ldots,2,1,2,\ldots,2)\rangle$,
where $d_H(v_1)$ is a vector of dimension $k/2$, such that
the $\frac{n_1}{2}$-th coordinate is $1$,
i.e., $d_H(v_1)_{n_1/2}=1$, and
$w_H(v_2)=\langle d_H(v_2)\rangle=\langle(2,\ldots,2,1)\rangle$,
where $d_H(v_2)$ is also a vector of dimension $k/2$.
So
$$
\Delta=
\big(\tfrac{n_1}2-\tfrac k2\big)(n_0-1)+\big(-1+\tfrac k2\big)(n_3-1)
=\tfrac{n_1-k}2\cdot(n_0-1)+\tfrac{k-2}2\cdot(n_3-1).
$$
This is nonnegative since $n_0\le n_3$ and $k-2\ge k-n_1$.
Moreover, $\Delta=0$ if and only if $n_0=n_3$ and $n_1=2$.

{\bf Case~2:}
{\it $k$ is even and $n_1$ is odd.}
Then $n_2$ is even and
$w_H(v_1)=\langle d_H(v_1)\rangle=\langle(2,\ldots,2,1)\rangle$,
where $d_H(v_1)$ is a vector of dimension $k/2$, and
$w_H(v_2)=\langle d_H(v_2)\rangle=\langle(2,\ldots,2,1,2,\ldots,2)\rangle$,
where $d_H(v_2)$ is also a vector of dimension $k/2$,
in which $d_H(v_2)_{n_2/2}=1$.
So
$$
\Delta=
0\cdot(n_0-1)+\big(-1+\tfrac{n_2}2\big)(n_3-1)
=\tfrac{n_2-2}2\cdot(n_3-1).
$$
This is nonnegative since $n_2-2\ge 0$.
Moreover, $\Delta=0$ if and only if $n_2=2$, since $n_3\ge 2$.

{\bf Case~3:}
{\it $k$ is odd and $n_1$ is even.}
Then $n_2$ is even and
$w_H(v_1)=\langle d_H(v_1)\rangle=\langle (2,\ldots,2,1,2,\ldots,2,1)\rangle$,
where $d_H(v_1)$ is a vector of dimension $(k+1)/2$, such that
$d_H(v_1)_{n_1/2}=1$ and
$w_H(v_2)=\langle d_H(v_2)\rangle=\langle(2,\ldots,2,1,2,\ldots,2,1)\rangle$,
where $d_H(v_2)$ is also a vector of dimension $(k+1)/2$, in which
$d_H(v_2)_{n_2/2}=1$.
Since $k-1-n_1=n_2-2$, we have
$$
\Delta=\big(\tfrac{n_1}2-\tfrac{k-1}2\big)(n_0-1)+\big(-1+\tfrac{n_2}2\big)(n_3-1)
=\tfrac{n_2-2}2\cdot(n_3-n_0).
$$
This is nonnegative since $n_2\ge 2$ and $n_3\ge n_0$.
Moreover $\Delta=0$ if and only if $n_2=2$ or $n_3=n_0$.

{\bf Case~4:}
{\it Both $k$ and $n_1$ are odd.}
Then $n_2$ is odd and
$w_H(v_1)=\langle d_H(v_1)\rangle=w_H(v_2)=\langle d_H(v_2)\rangle=\langle(2,\ldots,2)\rangle$,
where both these vectors are of dimension $(k-1)/2$.
So
$$
\Delta=\big(\tfrac{-(k-1)}2+\tfrac{k+1}2\big)(n_0-1)+\big(-1+\tfrac{k+1}2\big)(n_3-1)
=(n_0-1)+\tfrac{k-1}2\cdot(n_3-1)>0.
$$

Now combining these cases with Lemma~{\ref{lem:KC}}, which states that
$W(H')\ge W(H)$ and the equality holds if and only if $H=H'$ (see Case~3),
yields the result.
\end{proof}

In the following lemma we consider chains of traversal blocks.

\begin{lemma}
\label{lema:chains}
Let $n>p$.
Let $G$ be a graph on $n$ vertices with $p$ blocks which has the maximum
Wiener index.
Moreover, suppose that
$G=H_0\circ_{v_1}H_1\circ_{v_2}\dots\circ_{v_{\ell-1}}H_{\ell-1}\circ_{v_{\ell}}H_{\ell}$,
where $\ell\ge 2$, all $H_0,\dots,H_{\ell-1}$ are blocks and $H_{\ell}$
is a connected union of blocks.
Then $|V(H_1)|=\dots=|V(H_{\ell-2})|=2$.
Moreover, if $H_{\ell}$ is a terminal block or if
$|V(H_0\circ\dots\circ H_{\ell-2})|\le|V(H_{\ell})|$, then
$|V(H_{\ell-1})|=2$ as well.
\end{lemma}

\begin{proof}
Since $H_0$ is a terminal block and $H_1,\dots,H_{\ell-1}$ are traversal,
each of these blocks is either a cycle or $K_2$, by
Lemmas~{\ref{lem:terminal-blocks}} and~{\ref{lem:traversal-blocks}}.
Moreover, by Lemma~{\ref{lem:traversal-blocks}} we know that the attachment
vertices $v_i$ and $v_{i+1}$ are opposite on $H_i$, $1\le i\le\ell-1$.

Suppose that among $H_1,\dots,H_{\ell-2}$ there is a cycle on at least $3$
vertices, say $H_i$.
By Lemma~{\ref{lem:replacing-two-cycles}} both $H_{i-1}$ and $H_{i+1}$ must
be isomorphic to $K_2$.
Denote $t_1=|V(H_0\circ\dots\circ H_{i-1})|$ and
$t_2=|V(H_{i+1}\circ\dots\circ H_{\ell})|$.
We distinguish two cases.

{\bf Case~1:}
$t_1\le t_2$.
Denote $G_0=H_0\circ\dots\circ H_{i-2}$, $G_1=H_{i-1}$, $G_2=H_i$ and
$G_3=H_{i+1}\circ\dots\circ H_{\ell}$.
(Observe that if $i=1$ then $G_0$ is trivial consisting of a single vertex.)
Then $n_0=t_1-1<t_2=n_3$.
Hence, by Lemma~{\ref{lem:replacing-two-cycles}} we have $n_2=2$, a
contradiction.

{\bf Case~2:}
$t_1>t_2$.
Denote $G_0=H_{\ell}\circ\dots\circ H_{i+2}$, $G_1=H_{i+1}$, $G_2=H_i$ and
$G_3=H_{i-1}\circ\dots\circ H_0$.
Then $n_0=t_2-1<t_1=n_3$.
Hence, by Lemma~{\ref{lem:replacing-two-cycles}} we have $n_2=2$, a
contradiction.

Now we consider $H_{\ell-1}$.
If $|V(H_0\circ\dots\circ H_{\ell-2})|\le|V(H_{\ell})|$, then denote 
$G_0=H_0\circ\dots\circ H_{\ell-3}$, $G_1=H_{\ell-2}$, $G_2=H_{\ell-1}$ and
$G_3=H_{\ell}$.
(Observe that if $\ell=2$ then $G_0$ is trivial.)
Since $n_0<n_3$, by Lemma~{\ref{lem:replacing-two-cycles}} we have $n_2=2$.

If $H_{\ell}$ is a terminal block and $\ell\ge 3$, then relabelling the
blocks (reversing their order) we can prove that $|V(H_{\ell-1})|=2$.

Finally, if $H_{\ell}$ is a terminal block, $\ell=2$ and
$|V(H_0)|>|V(H_2)|$, then let $G_0$ be trivial, $G_1=H_2$, $G_2=H_1$ and
$G_3=H_0$.
Then $n_0<n_3$, and so $n_2=2$ by Lemma~{\ref{lem:replacing-two-cycles}}.
\end{proof}

By $\Theta_{a,b,c}$ we denote a graph consisting of two vertices,
which are connected by three internally vertex-disjoint paths of lengths $a$,
$b$ and $c$.
Observe that $\Theta_{a,b,c}$ has $a+b+c-1$ vertices.
In \cite[Lemma~5]{KS} we have the following statement.

\begin{theorem}
\label{thm:n+1}
Let $G$ be a $2$-connected graph on $n$ vertices,
having three vertices $v_1$, $v_2$ and $v_3$ such that
$$
D=\sum_{1\le i<j\le 3}d_G(v_i,v_j)
$$
is maximum  possible.
Then $D\le n+1$ and the equality is attained only if $G$ is $\Theta_{a,b,c}$,
where all $a$, $b$ and $c$ are even.
\end{theorem}

Observe that if $n$ is even then $D<n+1$ by Theorem~{\ref{thm:n+1}}.
Using this statement we prove the following lemma.

\begin{lemma}
\label{lema:2-terminal}
Let $n>p$.
Let $G$ be a graph on $n$ vertices with $p$ blocks which has the maximum
Wiener index.
Then $G$ has exactly two terminal blocks.
\end{lemma}

\begin{proof}
By Lemma~{\ref{lema:<4}}, $G$ has at most three terminal blocks.
By way of contradiction, suppose that $G$ has exactly three terminal blocks.
Then its blocks-tree has one vertex of degree $3$, three vertices of degree $1$
corresponding to terminal blocks, and all the remaining vertices have degree $2$.
The vertex of degree $3$ corresponds either to a~block or to a~cut-vertex.
To simplify the reasoning, in the latter case we consider the cut-vertex as
a~trivial block.

Hence, $G$ consists of a block $G_0$ with three vertices $v_1$, $v_2$
and $v_3$ in which there are attached connected unions of blocks
$G_1$, $G_2$ and $G_3$, respectively (obviously, the vertices $v_1$,
$v_2$ and $v_3$ are not necessarily disjoint).  We assume that $n_1\ge
n_2\ge n_3$.  By Lemma~{\ref{lema:chains}}, since $n_1\ge n_i$ for
$i\in\{2, 3\}$, we have $G_i=C_{k_i}\circ_{u_i} P_{t_i}$, where $u_i$
is one endvertex of the path $P_{t_i}$ and $v_i$ is another one, and
$k_i+t_i-1=n_i$.  Observe that $G_1$ may consist of two cycles
connected by a path, but we do not need to consider the structure of
$G_1$.  The structure of $G$ is visualized on Figure~{\ref{fig:1}}.

%
%

\begin{figure}[h]
\begin{center}
\includegraphics[scale=1.1]{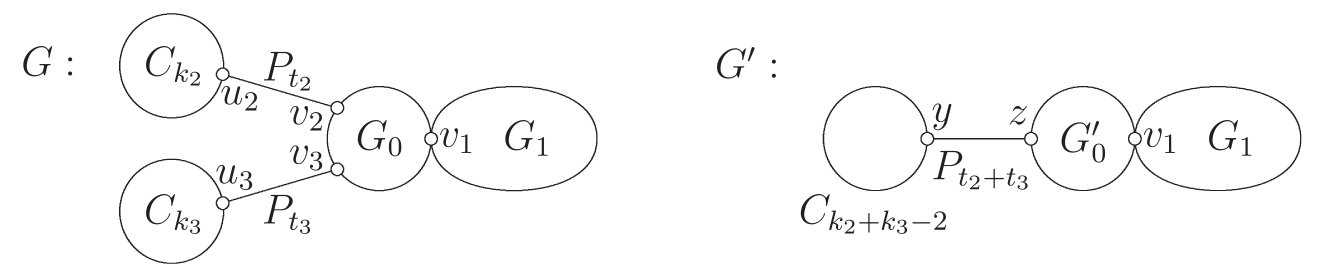}
\caption{Graphs $G$ and $G'$ in the proof of Lemma~{\ref{lema:2-terminal}}.}
\label{fig:1}
\end{center}
\end{figure}


Now we construct $G'$ on $n$ vertices with $p$ blocks,
so that $G'$ will have just two terminal blocks and $W(G')>W(G)$.
First, if $n_0\ge 3$, then let $G'_0$ be a cycle on $n_0$ vertices
in which $v_1$ is opposite to $z$.
If $n_0\le 2$, then $G_0$ is a cut vertex since the case $G_0=K_2$ is
impossible.
So set $G_0'=G_0$ and $z=v_1$ if $n_0=1$.
Let $z$ and $y$ be the two endvertices of $P_{t_2+t_3}$.
Then $G'=C_{k_2+k_3-2}\circ_y P_{t_2+t_3}\circ_z G'_0\circ_{v_1}G_1$, see
Figure~{\ref{fig:1}}.
Observe that the graphs $G$ and $G'$ have the same number of blocks and they
have also the same number of vertices.

Since $G'$ is simpler than $G$, we calculate $W(G')$ exactly.
However, for $W(G)$ we use just an upper bound $W$.
Below we show that $W(G')-W>0$.
Since $W(G)\le W$, this implies that also $W(G')-W(G)>0$.

By Proposition~{\ref{prop:P_n}}, $w_{P_a}(x)=\frac{a^2-a}2$ if $x$ is
an endvertex of a path of length $a$.
But if $x$ is a vertex of $C_a$ then $w_{C_a}(x)=\frac{a^2}4$ if $a$ is even
and $W_{C_a}(x)=\frac{a^2-1}4$ if $a$ is odd,
see Proposition~{\ref{prop:C_n}}.
Therefore, we distinguish two cases according to the parity 
of $k_2+k_3-2$.
If $k_2+k_3-2$ is odd, then exactly one of $k_2$ and $k_3$ is odd as well.
Since we do not use the inequality $n_2\ge n_3$ in the proof, 
without loss of generality we may assume that $k_2$ is even and $k_3$ is odd
in this case.
If $k_2+k_3-2$ is even, then either both $k_2$ and $k_3$ are even or both
are odd.
However, since it suffices to find an upper bound $W$ on $W(G)$
such that $W(G')-W>0$, we use the upper bounds $\frac{k_2^2}4$ and
$\frac{k_3^2}4$ for $w_{C_{k_1}}(u_1)$ and $w_{C_{k_2}}(u_2)$, respectively,
in this case.

Now we bound $W(G')-W(G)$ using Proposition~{\ref{prop:W2}}.
The graph $G$ is composed of six parts $C_{k_2}$, $P_{t_2}$, $C_{k_3}$,
$P_{t_3}$, $G_0$ and $G_1$, see Figure~{\ref{fig:1}}.
Therefore we have 6 terms in the first sum of ({\ref{sum:W2}}),
$2\binom 62$ terms due to the first two products in the second sum of ({\ref{sum:W2}})
and $\binom 62-5$ terms due to the third product in the second sum of ({\ref{sum:W2}}).
This yields 46 terms due to $G$.
The graph $G'$ is composed of four parts $C_{k_2+k_3-2}$, $P_{t_2+t_3}$,
$G_0'$ and $G_1$, see Figure~{\ref{fig:1}}.
Therefore we have 19 terms due to $G'$ in ({\ref{sum:W2}}).
Since there are too many terms, we divide them into several groups
and we show that the sum of terms in each group is nonnegative.

{\bf 1}.
First consider {\it the terms containing $w_{G_1}(v_1)$}.
These terms occur in the first two products of the second sum of
({\ref{sum:W2}}).
In $W(G)$ these terms are $(k_2-1)w_{G_1}(v_1)$, $(t_2-1)w_{G_1}(v_1)$,
$(k_3-1)w_{G_1}(v_1)$, $(t_3-1)w_{G_1}(v_1)$ and $(n_0-1)w_{G_1}(v_1)$.
Observe that they sum to $(n-n_1)w_{G_1}(v_1)$.
Since in $W(G')$ the three terms $(k_2+k_3-3)w_{G_1}(v_1)$,
$(t_2+t_3-1)w_{G_1}(v_1)$ and $(n_0-1)w_{G_1}(v_1)$ containing
$w_{G_1}(v_1)$ sum again to $(n-n_1)w_{G_1}(v_1)$,
these terms contribute $0$ to $W(G')-W(G)$.

{\bf 2}.
Now consider {\it the terms containing $w_{G_0}(v_1)$,
$w_{G_0}(v_2)$, $w_{G_0}(v_3)$, $w_{G'_0}(v_1)$, and
$w_{G'_0}(z)$}.
Since $w_{G'_0}(v_1)=w_{G'_0}(z)$, in $W(G')$ these terms sum to
$(n-n_0)w_{G'_0}(v_1)$.
By Proposition~{\ref{prop:C_n}}, we have
$w_{G_0}(v_i)\le w_{G'_0}(v_1)$, $1\le i\le 3$.
Hence, the upper bound for the contribution of considered terms
to $W(G)$ is also $(n-n_0)w_{G'_0}(v_1)$.
Consequently, these terms contribute at least $0$ to $W(G')-W(G)$.

{\bf 3.}
Now consider {\it the terms containing $(n_0-1)$ which were not considered
in the groups 1. and 2. above, together with the terms containing distances
in $G_0$ and $G'_0$}.
We start with the case when $k_2+k_3-2$ is even.

First consider the terms containing $(n_0-1)$.
Their contribution to $W(G')-W(G)$ is at least (the fractions correspond
to $w_H(x)$'s, while the non-fractions correspond to the last term
in ({\ref{sum:W2}})).
\begin{eqnarray}
&&(n_0-1)\Big[\tfrac{(k_2+k_3-2)^2}4+\tfrac{(t_2+t_3)^2-(t_2+t_3)}2
+(k_2+k_3-3)(t_2+t_3-1)\nonumber\\
&&\quad -\tfrac{k_2^2}4-\tfrac{k_3^2}4-\tfrac{t_2^2-t_2}2-\tfrac{t_3^2-t_3}2
-(k_2-1)(t_2-1)-(k_3-1)(t_3-1)\Big]\nonumber\\
&&=(n_0-1)\big[\tfrac 12(k_2-2)(k_3-2)+(k_2-2)t_3+(k_3-2)t_2+t_2t_3\big].
\label{eq:31}
\end{eqnarray}
Since $n_0\ge 1$, $k_2\ge 2$, $k_3\ge 2$, $t_2\ge 1$ and $t_3\ge 1$,
the expression (\ref{eq:31}) is nonnegative.

Now consider the terms containing distances in $G_0$ and $G_0'$.
In $W(G)$ these terms sum to $(n_1-1)(n_2-1)d_{G_0}(v_1,v_2)$,
$(n_1-1)(n_3-1)d_{G_0}(v_1,v_3)$ and $(n_2-1)(n_3-1)d_{G_0}(v_2,v_3)$.
By Theorem~{\ref{thm:n+1}} we have
$d_{G_0}(v_1,v_2)+d_{G_0}(v_1,v_3)+d_{G_0}(v_2,v_3)\le n_0+1$.
Since $n_1\ge n_2$ and $n_1\ge n_3$, we obtain the biggest contribution if
$d_{G_0}(v_1,v_2)$ and $d_{G_0}(v_1,v_3)$ are maximum possible,
namely $\lfloor\frac{n_0}2\rfloor$.
Then $d_{G_0}(v_2,v_3)\le2$.
Hence, the contribution of these terms to $W(G)$ is at most
$$
\big\lfloor\tfrac{n_0}2\big\rfloor\big[(n_1-1)(n_2-1)+(n_1-1)(n_3-1)\big]
+2(n_2-1)(n_3-1),
$$
while the contribution of the terms containing $d_{G'_0}(v_1,z)$
to $W(G')$ is
$$
d_{G'_0}(v_1,z)(n_1-1)(n_2+n_3-2)
=\big\lfloor\tfrac{n_0}2\big\rfloor(n_1-1)(n_2+n_3-2).
$$
Consequently, the contribution of these terms to $W(G')-W(G)$ is at least
\begin{equation}
\label{eq:32}
-2(n_2-1)(n_3-1)=
-2\big[(k_2-2)(k_3-2)+(k_2-2)t_3+(k_3-2)t_2+t_2t_3\big].
\end{equation}
Our aim is to show that the sum of the right-hand sides of (\ref{eq:31}) and
(\ref{eq:32}) is nonnegative.
We consider five cases.

{\bf Case~1:}
$n_0\ge 5$.
Since the expression in brackets containing $k_2$, $k_3$, $t_2$ and $t_3$
in (\ref{eq:31}) is nonnegative, it suffices to show nonnegativity of the sum
of (\ref{eq:31}) and (\ref{eq:32}) for $n_0=5$.
Since this sum is
$$
2(k_2-2)t_3+2(k_3-2)t_2+2t_2t_3>0,
$$
the contribution of selected terms is nonnegative in this case.

{\bf Case~2:}
$n_0=1$.
In this case the considered distances in $G_0$ and $G'_0$ are $0$ as well as
$(n_0-1)$.
Hence, the contribution of selected terms is $0$ in this case.

{\bf Case~3:}
$n_0=2$.
This case is impossible, since if $G_0=K_2$ then the vertex of degree $3$
in the blocks-tree is a cut-vertex.

{\bf Case~4:}
$n_0=3$.
In this case we have $d_{G_0}(v_i,v_j)=1$, $1\le i<j\le 3$, and also
$d_{G'_0}(v_1,z)=1$.
Hence, the contribution of the terms based on distances is
$$
-1\big[(k_2-2)(k_3-2)+(k_2-2)t_3+(k_3-2)t_2+t_2t_3\big]
$$
and the total contribution of considered terms is
$$
(k_2-2)t_3+(k_3-2)t_2+t_2t_3>0.
$$

{\bf Case~5:}
$n_0=4$.
By Theorem~{\ref{thm:n+1}} we have
$d_{G_0}(v_1,v_2)+d_{G_0}(v_1,v_3)+d_{G_0}(v_2,v_3)\le n_0 = 4$. In this
case, the sum of the terms containing distances in $G_0$ and $G'_0$ is
non-negative.
So the considered terms contribute to $W(G')-W(G)$ by at least
$$
(n_0-1)\big[\tfrac 12(k_2-2)(k_3-2)+(k_2-2)t_3+(k_3-2)t_2+t_2t_3\big]>0.
$$

Summing up, the contribution of considered terms to $W(G')-W(G)$ is
at least $0$ if $k_2+k_3-2$ is even.
If $k_2+k_3-2$ is odd, the only changes consist in replacing
$\frac{(k_2+k_3-2)^2}4$ and $\frac{k_2^2}4$ by
$\frac{(k_2+k_3-2)^2-1}4$ and $\frac{k_2^2-1}4$, respectively.
Hence, we obtain exactly the same expressions as in the even case.

{\bf 4}.
Now we consider {\it the terms containing $(n_1-1)$, which were not
considered before}.
Again, we start with the case when $k_2+k_3-2$ is even.
The contribution of the terms containing $(n_1-1)$ is at least
(compare with (\ref{eq:31}))
\begin{eqnarray}
&&(n_1-1)\Big[\tfrac{(k_2+k_3-2)^2}4+\tfrac{(t_2+t_3)^2-(t_2+t_3)}2
+(k_2+k_3-3)(t_2+t_3-1)\nonumber\\
&&\quad -\tfrac{k_2^2}4-\tfrac{k_3^2}4-\tfrac{t_2^2-t_2}2-\tfrac{t_3^2-t_3}2
-(k_2-1)(t_2-1)-(k_3-1)(t_3-1)\big]\nonumber\\
&&=(n_1-1)\big[\tfrac 12(k_2-2)(k_3-2)+(k_2-2)t_3+(k_3-2)t_2+t_2t_3\big].\nonumber
\end{eqnarray}
Since the expression in brackets containing $k_2$, $k_3$, $t_2$ and $t_3$
is nonnegative, we can replace $(n_1-1)$ by a value which is not larger
than $(n_1-1)$ and we will  not increase the contribution of considered terms.
Since $(n_1-1)\ge (n_i-1)=(k_i+t_i-2)$, $2\le i\le 3$, we get
$(n_1-1)\ge\frac 12(k_2+k_3+t_2+t_3-4)$.
Hence, the contribution of considered terms is at least
\begin{equation}
\label{eq:4}
(k_2+k_3+t_2+t_3-4)\big[\tfrac 12(k_2-2)(k_3-2)+(k_2-2)t_3+(k_3-2)t_2+t_2t_3\big]
\end{equation}
if $k_2+k_3-2$ is even.
If $k_2+k_3-2$ is odd, we obtain the very same expression.

{\bf 5}.
Finally, we consider {\it the remaining terms, i.e., the terms which were not
considered in the groups 1.-4. above}.
Then we include the terms from (\ref{eq:4}) and we show that their sum is positive.
Again, we start with the case when $k_2 + k_3 - 2$ is even.

Since $W(G'_0)-W(G_0)\ge 0$, the terms from the first sum
of Proposition~{\ref{prop:W2}} contribute to $W(G')-W(G)$ by at least
\begin{equation}
\label{eq:51}
\tfrac{(k_2+k_3-2)^3}8+\tfrac{(t_2+t_3)^3-(t_2+t_3)}6
-\tfrac{k_2^3}8-\tfrac{k_3^3}8-\tfrac{t_2^3-t_2}6-\tfrac{t_3^3-t_3}6.
\end{equation}

The terms from the second sum of (\ref{sum:W2})
contribute to $W(G')-W(G)$ by at least
\begin{eqnarray}
&&\quad\tfrac{(k_2+k_3-2)^2}4(t_2+t_3-1)+\tfrac{(t_2+t_3)^2-(t_2+t_3)}2(k_2+k_3-3)\nonumber\\
&&-\,\tfrac{k_2^2}4(k_3+t_2+t_3-3)-\tfrac{k_3^2}4(k_2+t_2+t_3-3)
-\tfrac{t_2^2-t_2}2(k_2+k_3+t_3-3)\nonumber\\
&&-\,\tfrac{t_3^2-t_3}2(k_2+k_3+t_2-3)-(k_2-1)(k_3-1)(t_2+t_3-2)\nonumber\\
&&-\,(k_2-1)(t_3-1)(t_2-1)-(k_3-1)(t_2-1)(t_3-1).
\label{eq:52}
\end{eqnarray}

And summing (\ref{eq:4}), (\ref{eq:51}) and (\ref{eq:52}) we get
\begin{eqnarray}
&&\quad\Big(\tfrac{k_3-2}4+\tfrac{t_3-1}2+\tfrac 12\Big)(k_2-2)^2
+\Big(\tfrac{k_2-2}4+\tfrac{t_2-1}2+\tfrac 12\Big)(k_3-2)^2\nonumber\\
&&+\Big(\tfrac{k_3-2}4+\tfrac{t_3-1}2\Big)(t_2-1)^2
+\Big(\tfrac{k_2-2}4+\tfrac{t_2-1}2\Big)(t_3-1)^2\nonumber\\
&&+\Big(\tfrac{k_3^2-4}8+\tfrac{k_3t_3}4+t_2t_3\Big)(k_2-2)
+\Big(\tfrac{k_2^2-4}8+\tfrac{k_2t_2}4+t_2t_3\Big)(k_3-2)\nonumber\\
&&+\big(t_2^2-1\big)\tfrac{k_3}4+\big(t_3^2-1\big)\tfrac{k_2}4
+2(t_2-1)(t_3-1)+\tfrac{t_2}2+\tfrac{t_3}2
\label{eq:5f}
\end{eqnarray}
which is positive since all the terms are nonnegative while the last two are
at least $\frac 12$ each.

Now consider the case when $k_2+k_3-2$ is odd.
Then (\ref{eq:4}) is without a change, (\ref{eq:51}) is increased by
$-\frac{k_2+k_3-2}8+\frac{k_3}8=\frac{2-k_2}8$ and (\ref{eq:52}) is increased
by $-\frac 14(t_2+t_3-1)+\frac 14(k_2+t_2+t_3-3)=\frac{k_2-2}4$.
So the sum of considered terms is exactly as in (\ref{eq:5f}) plus
a nonnegative term $\frac{k_2-2}8$.

Since all the groups of terms are nonnegative and the last one is positive,
the lemma is proved.
\end{proof}

Now combining Lemmas~{\ref{lema:chains}} and~{\ref{lema:2-terminal}}
we obtain Theorem~{\ref{thm:main}}.

\vskip 1pc
\noindent{\bf Acknowledgements.}~~The third and fourth authors acknowledge
partial support by Slovak research grants APVV-15-0220, APVV-17-0428,
VEGA 1/0142/17 and VEGA 1/0238/19.
The research was partially supported by Slovenian research agency ARRS,
program no. P1-0383.
The fifth author acknowledges partial support by National Scholarschip
Programme of the Slovak Republic SAIA.

%
%

\end{document}